%% file: main.tex
\documentclass[a4paper,USenglish,thm-restate,pdfa]{lipics-v2021}
\nolinenumbers

\usepackage[utf8]{inputenc}
\usepackage[ruled,vlined]{algorithm2e}
\usepackage{color}
\usepackage{amsmath,amsthm}
\usepackage{hyperref}

\usepackage[colorinlistoftodos,textsize=small,color=red!25!white,obeyFinal]{todonotes}

\newcommand{\mdnoteinline}[1]{\todo[inline, size=\normalsize, color=blue!20]{Mike's Note: #1}}

\title{Smoothed Analysis of Information Spreading in Dynamic Networks}

\author{Michael Dinitz}{Johns Hopkins University, Baltimore, MD, USA}{mdinitz@cs.jhu.edu}{}{Supported in part by NSF award CCF-1909111}
\author{Jeremy Fineman}{Georgetown University, Washington, DC, USA}{jfineman@cs.georgetown.edu}{}{Supported in part by NSF grants CCF-1918989 and CCF-2106759.}
\author{Seth Gilbert}{National University of Singapore, Singapore}{seth.gilbert@comp.nus.edu.sg}{}{Supported in part by Singapore MOE grant MOE2018-T2-1-160.}
\author{Calvin Newport}{Georgetown University, Washington, DC, USA}{cnewport@cs.georgetown.edu}{}{}

\authorrunning{M.\ Dinitz, J.\ Fineman, S.\ Gilbert, and C.\ Newport}

\Copyright{Michael Dinitz, Jeremy Fineman, Seth Gilbert, and Calvin Newport}



\EventEditors{Christian Scheideler}
\EventNoEds{1}
\EventLongTitle{36th International Symposium on Distributed Computing (DISC 2022)}
\EventShortTitle{DISC 2022}
\EventAcronym{DISC}
\EventYear{2022}
\EventDate{October 25--27, 2022}
\EventLocation{Augusta, Georgia, USA}
\EventLogo{}
\SeriesVolume{246}
\ArticleNo{37}

\begin{document}
\hideLIPIcs

\maketitle

\begin{abstract}
    The best known solutions for $k$-message broadcast in dynamic networks of size $n$
    require $\Omega(nk)$ rounds.
    In this paper, we see if these bounds can be improved by smoothed analysis.
    To do so, we study perhaps the most natural randomized algorithm for disseminating
    tokens in this setting: at every time step, choose a token to broadcast
    randomly from the set of tokens you know.
    We show that  with even a small amount of smoothing (i.e., one random edge
    added per round), this natural strategy solves $k$-message broadcast
    in $\tilde{O}(n+k^3)$ rounds, with high probability,
    beating the best known bounds for $k=o(\sqrt{n})$ and matching
    the $\Omega(n+k)$ lower bound for static networks for $k=O(n^{1/3})$ (ignoring
    logarithmic factors).
    In fact, the main result we show is even stronger and more general:
    given $\ell$-smoothing (i.e., $\ell$ random edges added per round),
    this simple strategy terminates in $O(kn^{2/3}\log^{1/3}(n)\ell^{-1/3})$ rounds.
    We then prove this analysis close to tight with an almost-matching lower bound.
    To better understand the impact of smoothing on information spreading,
    we next turn our attention to static networks,
    proving a tight bound of $\tilde{O}(k\sqrt{n})$ rounds to solve $k$-message broadcast,
    which is better than what our strategy can achieve in the dynamic setting.
    This confirms the intuition that although smoothed analysis reduces the difficulties
    induced by changing graph structures, it does not eliminate them altogether.
    Finally, we apply tools developed to support our smoothed analysis
    to prove an optimal result for $k$-message broadcast in so-called well-mixed
    networks in the absence of smoothing.  By comparing this result to an existing
    lower bound for well-mixed networks, we establish
      a formal separation between oblivious and strongly adaptive adversaries with respect to well-mixed token spreading,
     partially resolving an open question on the impact of adversary strength on the $k$-message broadcast problem.
\end{abstract}


\input{frontmatter}


\input{pred}

\section{Random Broadcast in Worst-Case Networks}
\label{sec:worstcase}

We begin by showing that in the absence of any additional assumptions,
random broadcast matches the best known bound of $O(nk)$ rounds for solving
$k$-message broadcast.
The sections that follow will then improve on this baseline.
The intuition for this result is straightforward: if token $t$
is not fully spread by round $r$ there is at least one edge over which it
could spread with probability at least $1/k$.
A union bound and stochastic dominance argument are deployed in the following 
proof to dispatch dependency and varying token set size issues, respectively.

\begin{theorem}
Fix a dynamic network ${\mathcal G}$ of size $n$. Fix a rumor set $T$ of size $k \leq n$.
With high probability: random broadcast solves $k$-message broadcast in $O(nk)$ rounds
in ${\mathcal G}$.
\label{thm:basic}
\end{theorem}
\begin{proof}
Fix a token $t$ and round $r$.
If $n_t(r) < n$ then there is at least one edge in the network
crossing the cut between nodes that do and nodes that do not know $t$.
Token $t$ is selected for broadcast across this edge with probability at least $1/k$.
Let $X_r$ be the random variable that evaluates to $1$ under the following two conditions:
(1) $n_t(r) = n$; or (2) $n_t(r) < n$ and at least one new node learns $t$.

Clearly, regardless of the execution,
for every $r$, $\Pr(X_r = 1) \geq 1/k$.
Let $Y=\sum_{r=1}^{\alpha n k} X_r$, for a constant $\alpha \geq 1$ that we will fix later.
We can apply Lemma~\ref{lem:stochasticdominance} for $p=1/k$, $j=\alpha n k$,
and $\varepsilon = 1/2$ to derive the following:

\[ \Pr[Y \leq (1/2)\cdot(1/k) \cdot (\alpha n k)] \leq \exp\left(-(1/8)\cdot \alpha n\right).   \]

Notice, due to the large size of the expectation,
for $\alpha \geq 8$, this bound gives us a failure probability exponentially small in $n$.
Clearly then, there is a sufficiently large constant $\alpha$ such that this failure probability
is less than or equal $n^{-2}$, allowing us to apply a union bound over all $k\leq n$ tokens
to establish that with high probability: {\em every} token spreads to all nodes in $\alpha n k$ rounds.
\end{proof}

\section{Random Broadcast in Smoothed Networks}
We now consider the random broadcast algorithm when the initial token distributions
and network topologies are arbitrary. In this worst-case setting,
as mentioned,
the best known $k$-message broadcast solutions require $nk$ rounds,
a bound which is known to be tight within log-factors under certain adversary assumptions.  In the previous section, we showed that random broadcast matched this bound as well.  The goal of this section is to analyze random broadcast under smoothed analysis. 

We show here that if we run the simple random broadcast algorithm on an $\ell$-smoothed dynamic network (with $\ell > 0$) then with high probability it solves $k$-message
broadcast after only $O\left(\frac{kn^{2/3}\log^{1/3} n}{\ell^{1/3}}\right)$ rounds.  Note that even in the case of little smoothing, e.g., $\ell=1$, this bound represents nearly an $n^{1/3}$-factor improvement to the round complexity captured by the worst-case bound.
To simplify comparison to existing $k$-message broadcast results,
 as well as the offline result proved later in this paper,
 we also prove that for $\ell=1$, this complexity is upper bounded
 by $O(n+k^3\log{n})$.
 
 Finally, we note that following the approach of~\cite{dinitz2018smoothed},
 we study only integer $\ell$ values.
 As recently demonstrated in~\cite{MPS20},
 fractional smoothing parameters, $0 < \ell < 1$, can also be studied
 to capture behavior in long-lived networks with slow changes.
 Our results naturally extend to this case, though we omit this analysis for 
 the sake of clarity.



\input{lSmoothed}

\input{lower}

\input{static}










\input{wellmixed}




\bibliography{smoothing}
\bibliographystyle{plainurl}
\end{document}

%% file: frontmatter.tex
\section{Introduction}
In this paper, we apply smoothed analysis to the study of $k$-message broadcast in dynamic networks.
We prove that even with a small amount of smoothing, a simple distributed random broadcast strategy can significantly
outperform the existing worst-case lower bounds. 
We then prove that in static networks the complexity of this strategy further improves, 
establishing that even in the context of smoothing, changing topologies remain more difficult to move information through
than their static counterparts.
Finally, we apply the tools developed for these analyses to improve the best-known bounds for $k$-message
broadcast, without smoothing, in the  {\em well-mixed} dynamic network setting. This result is significant in part because
when combined with an existing lower bound on well-mixed networks~\cite{dutta2013complexity},
it provides a formal separation between strongly adaptive and oblivious adversaries for $k$-message broadcast.

\subsection{Background}

In studying distributed network algorithms, it is common to represent the underlying topology with a graph,
where nodes correspond to processes and edges to communication links.
In the {\em dynamic network} setting, these graphs can change from round to round as determined by an adversary.
An upper bound proved in a dynamic network is considered strong as it can tolerate the many sources of interference, failure or congestion
that alter link availability in real world networks (see~\cite{dynamic-overview} for a good review).

Kuhn et~al.~\cite{kuhn2010distributed} sparked recent interest in the study of the {\em $k$-message broadcast problem}, in which nodes
in a network of size $n$ must spread $k$ messages (also called {\em tokens}) to the whole network.
In~\cite{kuhn2010distributed}, the results assume 
the Broadcast CONGEST model in which in each round,
each node can broadcast a single bounded-size message,
containing at most $1$ token.
A primary result in the paper
is a deterministic algorithm that solves $k$-message
broadcast in $O(nk)$ rounds.
For larger values of $k$, this is notably slower than 
the $O(n + k)$ rounds required to solve
this problem in a static network, underscoring the difficulty of dynamic topologies. 

Follow-up work by Dutta et al.~\cite{dutta2013complexity} proved this result close to optimal
with a lower bound that establishes $\Omega(nk/\log{n} + n)$ rounds are necessary to solve
$k$-message broadcast in this setting.
This result is strong in that it holds even for randomized algorithms (with a strongly adaptive adversary), and under the {\em well-mixed} token assumption in which each token has independent constant probability
of starting at each node.

\subsection{Key question: is $\tilde{\Omega}(nk)$ fundamental?}

Given the importance of information dissemination, an $\tilde{\Omega}(nk)$ lower bound
on $k$-message broadcast is unfortunately strong, 
especially for large networks attempting to disseminate large amounts of information in a setting with limited bandwidth.
Following the approach of Dinitz et~al.~\cite{dinitz2018smoothed}, however,
we can investigate whether this bound is fundamental.

In more detail, there are two useful possibilities to consider here.
First, this $\tilde{\Omega}(nk)$ bound might be {\em robust} in the sense that something like $nk$ rounds
to broadcast $k$ messages is a natural consequence of network topologies that change.
This would be reflected, for example, in the existence of large classes of graphs in which
this bound is obviously unavoidable.
The second possibility is that the bound is instead {\em fragile}
in the sense that it requires carefully-crafted pathological
topologies to induce a complexity of this magnitude,
and even small changes to these worst-case graphs enable much more efficient solutions.
These distinctions are important because if the $\tilde{\Omega}(nk)$ lower bound due to~\cite{dutta2013complexity} can be shown to be fragile,
this provides hope that more efficient information dissemination can be expected in most real world settings. By contrast, if the bound is robust, this indicates that efficient communication should not be expected in practice.

One approach to distinguishing between robustness and fragility is to apply smoothed analysis.
In more detail, in the study of sequential algorithms,
Spielman and Teng introduced {\em smoothed analysis}
to help explain why the simplex algorithms works well in practice
despite pessimistic worst-case lower bounds~\cite{SpielmanT04,SpielmanT09}.
They proved that the introduction of small random perturbations
to otherwise worst-case inputs enabled stronger bounds,
indicating the existing lower bound was fragile.

Dinitz et al.~\cite{dinitz2018smoothed}
subsequently adapted smoothed analysis to the study of distributed
algorithms in dynamic networks.
In this framework, as in the worst-case setting,
an adversary generates an arbitrary dynamic graph to describe
the changing network. The individual graphs, however,
are then each augmented 
with $\ell$ additional random edges, for some smoothing parameter $\ell$,
before the distributed algorithm in question is run.

For $\ell=0$, this reduces to the standard worst-case setting where existing lower bounds apply.
For $\ell = \binom{n}{2}$, this reduces (more or less) to a random graph setting,  in which much stronger upper bounds results are typically possible.
As argued in~\cite{dinitz2018smoothed}, if a worst case lower bound is significantly diminished
by smoothed analysis for small $\ell$ values, then this hints that the original bound is fragile.

The processes and problems studied in~\cite{dinitz2018smoothed} were flooding, random walks, and token aggregation. (Follow-up work applied smoothed analysis to the study of the minimum spanning tree~\cite{chatterjee2020distributed} and
leader election~\cite{molla2020smoothed} in static graphs.)
The $k$-message broadcast problem features arguably the best-known pessimistic lower bound in the dynamic network setting, but its examination using smoothed analysis
was left in~\cite{dinitz2018smoothed} as an open problem.

\subsection{Our Results}
We focus in this paper on {\em random broadcast}, one of the simplest possible algorithms
for disseminating tokens: in each round, each node broadcasts a token
chosen uniformly at random from its current token set.  
This simple strategy will enable us to prove a variety of interesting results on $k$-message broadcast.

\subsubsection*{Smoothed analysis of random broadcast in dynamic networks}

Applying smoothed analysis as our key tool, 
we prove that even a small amount of smoothing (i.e., one random edge added per round)
is sufficient to enable random broadcast to outperform the worst-case lower bound.
This implies both that the existing bound is fragile, and that random broadcast
is, in some sense, the {\em right} strategy for spreading tokens through a dynamic network.

We first establish in Section~\ref{sec:worstcase}
the baseline result that with no smoothing random broadcast solves the problem in $O(nk)$ rounds, with high probability in $n$.  This matches the deterministic bound from~\cite{kuhn2010distributed}.\footnote{Note that~\cite{kuhn2010distributed} gave a deterministic algorithm for the problem, and also explored the problem of termination detection,
which further complicates the problem.} 
We then investigate the impact of smoothing. In Section~\ref{sec:Lsmoothing},
we show that even with  a small amount of smoothing (i.e., $\ell = 1$), 
random broadcast now terminates in $\tilde{O}(n + k^3)$ rounds, with high probability in $n$,
improving on the best-known $O(nk)$ bound for any $k = \tilde{o}(\sqrt{n})$,
and matching the static network lower bound of $\Omega(n+k)$ for $k=\tilde{O}(n^{1/3})$.
 
We emphasize that $1$-smoothing adds at most \emph{one} new edge to the network in each round,
which enables at most one extra token dissemination. This smoothing therefore changes the overall bandwidth or connectivity of the graph by only a very small amount.\footnote{Notice, for example, that to significantly increase the conductance or vertex expansion of the graph, you would need to add many more edges.}
Given that $\Theta(nk)$ rounds might be necessary to solve $k$-message broadcast,
our speed-up in time complexity in this context
does not come simply from adding large amounts of extra capacity to a worst-case network: most of the work of token dissemination must still occur over the adversarially specified edges in the network.
The smoothing   accomplishes something more subtle: as
we elaborate in our below discussion of predecessor paths,
these extra edges are not eliminating bottlenecks in the underlying dynamic network, 
but instead providing just enough random noise to allow us to bypass their corresponding potential for congestion.

In reality, our result for $\ell=1$ is a special case of our more general result,
showing that random broadcast solves $k$-message broadcast in $O(\frac{kn^{2/3}\log^{1/3}{n}}{\ell^{1/3}})$ rounds,
which for $\ell=1$ is upper bounded by the $\tilde{O}(n+k^3)$ bound claimed above for all $n$ and $k$.
Notice that in this general form, for $k=o(n^{1/3})$,
random broadcast actually {\em beats} the $\Omega(n+k)$ lower bound for static networks.
This is possible because even a small number of additional random edges
enables tokens to not only bypass bottlenecks,
but also skip ahead in temporal paths, reducing the effective dynamic diameter of the network.
As we increase $\ell$, we get further improvements. 
For $\ell=k^3$, for example,
we get a sub-linear result,
as the increased smoothing both speeds up the rate at which tokens
initially spread, and the rate at which they subsequently
jump over smoothed edges to locations near their destinations in time and space (see below for elaboration).

\subparagraph*{Predecessor paths.} A key technique in our analysis,
presented in Section~\ref{sec:pred},
is the use of graph structures that we call {\em predecessor paths},
which capture paths that exist over time.
They are represented as a sequence of node/round pairs,
$(u_1,r_1),(u_2,r_2), \ldots, (u_x, r_x)$,
and for each $(u_i, r_i)$, for $i< x$,
it is guaranteed that $u_i$ is connected to $u_{i+1}$
during round $r_i$.
We further customize these paths for a given token $t$,
strengthening the guarantee for each $(u_i,r_1)$ such that not only will 
$u_i$ be connected to $u_{i+1}$ in round $r_i$,
but it will broadcast token $t$ in this round, if it knows it.

For each given destination $u_x$ and token $t$,
therefore, we can reduce the problem of delivering $t$ to $u_x$ to the problem
of seeding token $t$ into the appropriate predecessor path. 
(Intuitively, we are establishing here a net over time and space that can capture
a token and then inexorably guide it to the center of the trap.)
At a high-level, we can therefore break our smoothed analysis of random broadcast
into three phases. During the first phase, we ignore the smoothed edges,
and allow the natural dynamics of information spread in these networks
spread out each token to a larger set.
During the second phase, we allow the smoothed edges to seed these
tokens onto the appropriate predecessor paths.
During the final phase, the tokens can then traverse these paths to their final destinations.

The lengths of these phases are inter-dependent. 
Increasing the initial spreading phase, for example,
decreases the second phase as now each smoothed edge has a higher probability of selecting
a node with a useful token.
Similarly,
relying on long predecessor paths also reduces the second phase,
as now each smoothed edge has more targets to which to deliver a token.
Longer predecessor paths, however, necessitate a longer third phase to given tokens time to traverse to their destinations.
Our final result balances these dependencies by optimizing the time complexity
when we fix all three phases to be the same length.

\subparagraph*{Lower bound.}  In Section~\ref{sec:lower}, we complement our upper bound analysis of random broadcast
with a nearly-matching lower bound.
In more detail
we describe and analyze a {\em dynamic star} topology
in which the network graph forms a star in each round, but the identity of the center node rotates over time.
In this setting, we prove random broadcast's expected time to solve $k$-message
broadcast is in $\tilde \Omega\left( \min\left( \frac{kn^{2/3}}{(\ell(k+\ell))^{1/3}}, \frac{kn}{k+\ell}\right)\right)$. 
This result is approximately a factor of $(k+\ell)^{1/3}$ below our upper bound analysis,
confirming that significantly more efficient analyses are not possible.
Notably, this bound establishes the fundamental nature of the drop from $n$ to $n^{2/3}$ in the presence
of even a small amount of smoothing.

\subsubsection*{Smoothed analysis of random broadcast in static networks}

A possible interpretation of our upper bounds is the following:
``with a small amount of smoothing, dynamic networks behave like static networks''.  In other words, it might be the case that smoothing removes the differences between dynamic and static networks.  
In Section~\ref{sec:static}, we investigate this issue by studying the behavior of random broadcast in the network toplogies that do not change from
round to round.
We prove that in the presence of minimum smoothing (i.e., $\ell=1$)
random broadcast completes in static network a polynomial factor of $n$ faster than what is possible in dynamic networks.
Formally, we prove that in any static network with $1$-smoothing random broadcast completes in $\tilde{O}(k\sqrt{n})$ rounds, with high probability.
(Recall, the relevant upper bound result in dynamic networks for $1$-smoothing is $\tilde{O}(kn^{2/3})$ rounds.)
We then prove this analysis is tight (within logarithmic) terms with a matching lower bound.

At the core of our analysis is a  decomposition of an arbitrary static network into at most $O(\sqrt{n})$ components 
each with diameter at most $O(\sqrt{n})$.
We demonstrate that given a collection of $t$ components that know the token,
in a single spreading interval of $\tilde{O}(k\sqrt{n})$ rounds, each of the $t$ components is likely to send
a given target token over a smoothed edge to a unique new component, effectively doubling the number that now know it.
We leverage this doubling behavior to spread a token to a large fraction of the network in only a logarithmic number of spreading intervals.

\subsubsection*{Well-mixed networks}

In~\cite{dutta2013complexity}, the authors introduced the notion of a {\em well-mixed} network in
the context of studying $k$-message broadcast.
They call a network well-mixed if for every node $u$ and token $t$,
node $u$ starts with token $t$ with some independent constant probability.
Surprisingly, they observe that their $\Omega(nk)$ lower bound holds even for well-mixed networks. It turns out that even 
starting with a very uniform token distribution does not make the problem easy. 

Accordingly, they replace the broadcast communication model with the much more powerful
Symmetric-Diff CONGEST
model in which each node can not only send a different token on each outgoing edge,
but also perform a set-difference with each of their neighbors before deciding what tokens to send.
Given this extra power, they show that it is possible to solve $k$-message broadcast
in a well-mixed network
in $\tilde{O}(n+k)$ rounds, with high probability.

Leveraging our predecessor path constructions introduced for our smoothed analysis results,
we prove, perhaps surprisingly, that our simple random broadcast algorithm solves $k$-message
broadcast in $O((k/p)\log{n})$ rounds, with high probability, where $p$ is the probability that
each nodes starts with each token.
For the $p=\Theta(1)$ case considered in~\cite{dutta2013complexity}, 
we strictly improve on the bound they achieved in their more powerful communication model.
Indeed, for constant $p$, our bound is within a single log factor of matching a trivial
$\Omega(k)$ lower bound for {\em all} algorithms in the broadcast communication model.

The key follow-up question, of course, is why our result does not violate the $\Omega(nk)$ lower bound
from~\cite{dutta2013complexity}. The difference is found in the adversary assumptions.
The existing lower bound requires a {\em strongly adaptive} adversary that knows the nodes' random
choices in advance and can construct the network topology for a given round based
on the knowledge of the tokens nodes are broadcasting in that round.
We assume, by contrast, an oblivious adversary that designs the network without advance knowledge of these random bits.

Our well-mixed network result, therefore, opens a clear gap between the strongly adaptive and oblivious adversaries
in the context of $k$-message broadcast.
This partially resolves the open question presented in~\cite{dutta2013complexity}
as to whether or not the $\Omega(nk)$ lower bound applies to oblivious adversaries as well.

\section{Related Work}
\label{sec:related}

Many problems have been studied in various dynamic network models;
e.g.,~\cite{kuhn:2011,haeupler:2011,dutta2013complexity,clementi:2012,augustine:2012,denysyuk:2014,newport:2014,ghaffari:2013} (see~\cite{dynamic-overview} for a good survey).
Interest in $k$-message broadcast in a dynamic network with broadcast
communication was sparked by Kuhn et al.~\cite{kuhn:2010},
who established the original $O(nk)$ upper bound that provides the baseline
 for the smoothed analysis deployed in this paper.
The relevant matching lower bounds for arbitrary and well-mixed token
distributions were subsequently proved by Dutta et~al.~\cite{dutta2013complexity}.

Dinitz et al.~\cite{dinitz2018smoothed} adapted the smoothed analysis
technique, originally introduced by Spielman and Teng~\cite{SpielmanT04,SpielmanT09}
in the context of sequential algorithms, to dynamic networks.
They studied flooding, random walks and aggregation,
and identified $k$-message broadcast as an important open question.
Subsequent work applied this smoothed analysis framework to various
other graph problems, including 
minimum spanning tree construction~\cite{chatterjee2020distributed}
and leader election~\cite{molla2020smoothed}.
Recently, Meir et al.~\cite{MPS20} proposed a variation of graph
smoothing, suitable for long-lived processes, in which the smoothing parameter
$\ell$ can be fractional.
As noted in~\cite{dinitz2018smoothed},
smoothed analysis is not the only technique deployed
in the literature for sidestepping fragile dynamic
network lower bounds.
Denysyuk et al.~\cite{denysyuk:2014}, for example,
circumvent an exponential lower bound for random walks in dynamic
graphs due to~\cite{avin:2008}
by requiring the dynamic graph to include a certain number of static graphs from a well-defined set.
In the context of the dynamic radio network model,
Ghaffari et~al.~\cite{ghaffari:2013} studied the impact of adversary strength,
similarly finding a noticeable gap between oblivious and strongly adaptive
adversaries in the context of broadcast.

\section{Preliminaries}
\label{sec:prelim}
Here we define the dynamic network models we study and the $k$-message problem we solve.
We also formalize $\ell$-smoothing and establish some useful notation and
probability results leverage throughout the paper to follow.

\paragraph*{Model} 
We study a dynamic network model
in which an execution begins with an oblivious adversary that chooses a {\em dynamic graph,},
defined as a sequence ${\mathcal G} = G_1, G_2, \ldots$, where each $G_i$ is a connected graph over a common
node set $V$ of size $n=|V|$.
Time proceeds in synchronous rounds.
At the beginning of each round $r\geq 1$, each node $u\in V$ can reliably
broadcast a message to is neighbors in $G_r$. A key difficulty of these models is that $u$ does not
know its neighbors in advance.

\paragraph*{$k$-Message Broadcast} 
The $k$-message broadcast problem assumes a set $T$ containing $k\geq 1$ unique messages
that are also commonly called {\em tokens} or {\em rumors}.
Each rumor in $T$ starts the execution at one or more nodes.
The problem is solved once all nodes have received all $k$ rumors in $T$.
Following the standard convention~\cite{kuhn2010distributed},
we assume each node can broadcast at most $1$ rumor per round.

\paragraph*{$\ell$-Smoothing} 
Fix a dynamic graph ${\mathcal G}=G_1,G_2,\ldots$. 
Fix a smoothing parameter $\ell\geq 1$.
Our goal is to define a smoothing process that adds $\ell$ random edges
to each $G_i$ in ${\mathcal G}$.
In an effort
to maximize generality, 
the original definition of $\ell$-smoothing from~\cite{dinitz2018smoothed}
assumed that each $G_i$ was replaced by a graph $\hat G_i$ sampled
uniformly from the set of graphs that are both ``allowable'' and within edit distance $k$ of $G_i$.
This was meant to allow both additions and deletions while avoiding
illegal topologies (e.g., disconnected graphs). 

The results for the Broadcast CONGEST model in~\cite{dinitz2018smoothed},
however, largely avoided much of this generality,
instead applying results (Lemmas 4.1 and 4.2) that establish
that this model approximates a simpler model in which edges
are randomly added from the set of all edges.
For the sake of clarity, in this paper we directly
deploy this simpler definition of smoothing.

Formally, after the adversary generates ${\mathcal G}$,
we smooth each $G_i$ as follows: (1) randomly generate $\ell$ edges (with replacement);
(2) for each such edge, if it is not already in the smoothed graph we are generating,  add it to the graph.

\paragraph*{Notation}
In the following, we use $\tilde{O}$, $\tilde{\Theta}$, and $\tilde{\Omega}$
to suppress logarithmic factors with respect to $n$.
When we specify a result holds \emph{with high probability},
we mean with failure probability upper bounded by $n^{-x}$
for some sufficiently large constant $x\geq 1$.
We also use $[x]$, for integer $x\geq 1$, to indicate the set $\{1,2,\ldots,x\}$.

Fix an execution of a $k$-message broadcast algorithm for some token set $T$
of size $k$
and node set $V$.
For each node $u\in V$ and round $r\geq 1$,
let $T_u(r)$ be the set of tokens $u$ started with or received by
the beginning of round $r$.
We say $u$ {\em knows} the tokens in $T_u(r)$ at the beginning of round $r$.
Finally, for a given token $t\in T$,
let $n_t(r) = |\{ u \mid t\in T_u(r) \}|$ be the number of nodes
that know token $t$ at the beginning of $r$.

\paragraph*{Useful Probability Results}
Many of our high probability results that follow leverage the following
useful form of a Chernoff Bound:

\begin{theorem}
\label{thm:chernoff}

Let $X_1,\ldots,X_j$ be a series of independent random variables such that $X_i\in[0,1]$ where  $X =\sum_{i=1}^j X_i$ has expectation $E[X]=\mu$. For $\varepsilon\in [0,1]$, $\Pr[X\leq(1-\varepsilon)\cdot\mu]\leq\exp(-(1/2)\cdot\varepsilon^2\mu)$.

\end{theorem}

In several places in our analysis, we tame correlated random variables 
by applying the following stochastic dominance result,
which generalizes the above concentration bound.  It says that if the probability that the $i$'th variable is $1$ is at least $p$ no matter how the first $i-1$ variables are realized, then we can assume that we have independent Bernoulli variables with parameter $p$.  

\begin{lemma} \label{lem:stochasticdominance}
Let $X_1,\ldots,X_j$ be $j$ random variables (not necessarily independent), each of which is distributed over $\{0,1\}$.  Suppose there is some $p \in [0,1]$ such that for all $i \in [j]$ and for all $x_1, x_2, \dots, x_{i-1} \in \{0,1\}$,
\[
\Pr\left[X_i = 1 \mid X_k = x_k\ \forall 1 \leq k < i\right] \geq p.
\]
Then
\[
\Pr\left[\sum_{i=1}^j X_i \leq (1-\varepsilon)pj \right] \leq \exp(-(1/2)\cdot\varepsilon^2 p j)
\]
\end{lemma}

\begin{proof}
Consider the following process for sampling random variables $\hat X_1, \dots, \hat X_j$.  For $i = 1$ to $j$, do the following.  Let $x_1, \dots, x_{i-1} \in \{0,1\}$ be the values of $\hat X_1, \dots, \hat X_{i-1}$ respectively, and let $q = \Pr[X_i = 1 \mid X_k = x_k\ \forall 1 \leq k < i]$.  Note that by the assumption of the lemma, $q \geq p$.  Sample an independent Bernoulli random variable $Y_i$ which is $1$ with probability $p$ and is $0$ otherwise.  If $Y_i = 1$ then set $\hat X_i = 1$.  If $Y_i = 0$, then set $\hat X_i = 1$ with probability $(q - p)/(1-p)$ and otherwise set $\hat X_i = 0$.

It is easy to see that $\hat X_1, \dots, \hat X_j$  and $X_1, \dots, X_j$ have  identical joint distributions.  More formally, let $x_i \in \{0,1\}$ for all $i \in [j]$.  Then
\begin{align*}
    \Pr[\hat X_i = x_i\ \forall i \in [j]] &= \prod_{i=1}^j \Pr[ \hat X_i = x_i \mid \hat X_k = x_k\ \forall k < i] \\
    &= \prod_{i=1}^j \begin{cases}p + (1-p)\frac{\Pr[X_i = 1 \mid X_k = x_k\ \forall 1 \leq k < i] - p}{1-p} & \text{if } x_i = 1\\ (1-p) \frac{1 - \Pr[X_i = 1 \mid X_k = x_k\ \forall 1 \leq k < i]}{1-p} & \text{if } x_i = 0\end{cases} \\
    &= \prod_{i=1}^j \Pr[X_i = x_i \mid X_k = x_k\ \forall 1 \leq k < i] \\
    &= \Pr[X_i = x_i\ \forall i \in [j]]
\end{align*}

By the definition of $\hat X_i$, we also have the property that $Y_i \leq \hat X_i$ for all $i$.  Hence Theorem~\ref{thm:chernoff} implies that 
\begin{align*}
    \Pr\left[\sum_{i=1}^j X_i \leq (1-\varepsilon)pj \right] &= \Pr\left[\sum_{i=1}^j \hat X_i \leq (1-\varepsilon)pj \right] \leq \Pr\left[\sum_{i=1}^j Y_i \leq (1-\varepsilon)pj \right] \leq \exp(-(1/2)\cdot\varepsilon^2 p j)
\end{align*}
as claimed.
\end{proof}

%% file: pred.tex
\section{Random Broadcast Predecessor Paths}
\label{sec:pred}

Several of the results that follow are built on a structure  that we call {\em predecessor paths},
which are defined with respect to both a given dynamic graph ${\mathcal G}$ and the collection
of random bits that determine the choices during a given execution of the random broadcast algorithm.
We define and analyze these structures in a general way here. We will later deploy
these results to prove specific bounds on random broadcast.

To formally define a predecessor path, we first introduce the notion of a bit assignment ${\mathcal B}$
to be a function ${\mathcal B}:V \times \mathbb{Z}_{>0} \rightarrow \{0,1\}^*$,
where ${\mathcal B}(u,r)$ are the random bits node $u$ uses to make its choice of which
token to broadcast in round $r$ of running random broadcast.
Notice, the combination of a dynamic graph ${\mathcal G}$ and bit assignment ${\mathcal B}$, does not by itself 
fully specify an execution of random broadcast,
as knowledge of the initial token assignment is also required.
This information, however, is sufficient for our formal definition.\footnote{A brief
aside is that although we call these objects {\em paths}, the definition
actually captures a more general in-tree structure in the time expansion graph.}

\begin{definition}
Fix a dynamic graph ${\mathcal G}$ defined over node set $V$, token set $T$,
target node $u\in V$, target token $t\in T$, round pair $r$,$r'$,
with $1 \leq r < r'$, and bit assignment ${\mathcal B}$.
A {\em predecessor path} $P_{u,t}(r,r')$ for these parameters is a node/round
sequence $(u_1,r_1),(u_2,r_2),\ldots,(u_h,r_h)$, 
where $r \leq r_1 < r_2 < \ldots < r_h \leq r'$, 
and $u_i \neq u_j$ for $i,j \in [h]$, $i\neq j$, that
satisfies the following
 with respect to the execution of random broadcast in ${\mathcal G}$ according to bit assignment ${\mathcal B}$:

\begin{enumerate}
    \item For each $i\in [h-1]$: if $t\in T_{u_i}(r_i)$, 
    then $u_i$ will broadcast $t$ in round $r$
    and it will be received by at least one node $u_j$, for $j>i$.
    
    \item If $t\in T_{u_h}(r_h)$, then $u_h$ will broadcasts $t$ during round $r_h$ and it will
    be received by node $u$.
\end{enumerate}
\end{definition}

A natural corollary of this definition is that if {\em any} node $u_i$
in a predecessor path $P_{u,t}(r,r')$ learns $t$ by $r_i$,
then $u$ will learn $t$ by $r'$.

Our goal is to describe and analyze a procedure for generating
a predecessor path for a given set of parameters.
We will then analyze the expected length of the paths created.
Roughly speaking, when studying random broadcast,
the longer a predecessor path the better, as it gives more opportunities for a
given token to arrive at a node that can then send it on its way to the desired
destination.

\paragraph*{Preliminaries}
As discussed, we will be analyzing the simple algorithm in which every node broadcasts a token that it knows uniformly at random.  To ease the analysis, though, we assume without loss of generality that (1) the tokens are labelled from $1$ to $k$ and that nodes know $k$;
and (2) that a node randomly selects a token to send in a given round by randomly permuting
the values from $1$ to $k$ and then broadcasting the first token from this sequence that it possesses.  Note that this gives the exact same process as the randomized algorithm we care about, which is why this is without loss of generality. It allows us, however, to fix the random process for how a token is chosen even when the available set of tokens to send is unknown.

For a given node set $V$, token set $T$,  bit assignment ${\mathcal B}$,
and round $r\geq 1$,
let the {\em primary token} for $u$ in $r$, indicated $\delta_u(r)$, 
be the first token id in $u$'s random permutation for this round as determined 
by ${\mathcal B}$. In the practical setting where $u$ permutes all values from $1$ to $\hat k$,
then  the primary token is the first value in this permutation that correspond to an
actual token in $T$.
The important property of a primary token
is if $\delta_u(r) = t$,
then we know that if $t\in T_u(r)$, node $u$ will broadcast $t$ in this round.

Given a dynamic graph ${\mathcal G} = G_1,G_2, \ldots$,
a non-empty node subset $S\subset V$, and a round $r\geq 1$,
we further define the {\em predecessor cut} $c(S,r)$ to be
the set of nodes in $V \setminus S$ that neighbor nodes in $S$ in $G_r$.
Notice, because each graph is connected and $S$ is a proper subset of $V$,
these cut partitions are always non-empty.

\paragraph*{Predecessor Path Construction}
We now describe how to construct a predecessor path for a given set of parameters.
This construction process works backwards in time from the end of the desired
interval to the beginning.
While we describe this process algorithmically, we emphasize that this algorithmic construction is used only in the \emph{analysis} of our algorithms. 

In the following,
we assume a fixed dynamic network ${\mathcal G} = G_1,G_2,\ldots$, 
defined over some node set $V$ of size at least $2$,
a token set $T$ of size $k\geq 1$,
and a fixed bit assignment ${\mathcal B}$ for the nodes in $V$ to run random broadcast.
We then parameterize the construction process with a node $u\in V$,
token $t\in T$, and round range $1\leq r < r'$.
It returns a predecessor path $P_{u,t}(r,r')$ for these parameters.

\begin{algorithm}[H]
\SetAlgoLined
 $P_{u,t}(r,r') \gets \epsilon$\;
 $i \gets r'$\;
 $S\gets \{u\}$\;
 
 \While{$i > r$}{
    $S_{i-1} \gets c(S, (i-1))$\;
    $S_{i-1}^{(t)} \gets \{ v \mid v\in S_{i-1} \wedge \delta_v(i-1) = t\}$\;
    \If{$|S_{i-1}^{(t)}|>0$}
    {
    fix any $v$ in $S_{i-1}^{(t)}$\;
    $S \gets S \cup \{v\}$\;
    append $(v,i-1)$ to the front of $P_{u,t}(r,r')$\;
    }
    $i \gets i - 1$\;

 }
 
\Return{$P_{u,t}(r,r')$}
 \caption{Path-Construction($u,t,r,r')$}
\end{algorithm}

We now analyze the paths constructed by this procedure, showing establishing the relationship between interval length ($r'-r$) and path length.

\begin{theorem}
Fix a dynamic graph ${\mathcal G}$ defined over node set $V$ of size $n>1$,
token set $T$ of size $k\geq 1$, 
random bit assignment ${\mathcal B}$ for $V$,
and error exponent integer $x>0$.
For every $u\in V$, $t\in T$, and rounds $r,r'$ where $r'-r = z\geq 8xk\ln{n}$:
\begin{enumerate}
    \item The sequence $P_{u,t}(r,r')$ produced by {\em Path-Construction} is a predecessor path.
    \item With probability at least $1-n^{-x}$: $|P_{u,t}(r,r')| > \frac{z}{2k}$.
\end{enumerate}

%
%
%
\label{thm:pathlength}
\end{theorem}
\begin{proof}
Fix values for the parameters specified and constrained in the theorem statement. 
Consider the sequence $P_{u,r}(r,r')$ produced by the Path-Construction
algorithm for these parameters. 
By the definition of this algorithm,
$P_{u,t}(r,r')$ is a valid predecessor path.
We turn our attention to bounding its size.

At each iteration of the main loop in the  procedure,
if $|S| < n$, then
 $S_{i-1}$ is non-empty, as $V\setminus S$ is non-empty and $G_{i-1}$ is connected.
 Fix some $u\in S_{i-1}$.
 This node is included into $S_{i-1}^{(t)}$ only if $\delta_v(i-1) = t$,
which occurs with probability exactly $1/k$.
Therefore, the probability that our path expands in this iteration in the case that $|S|<n$
is at least $1/k$ (it could be larger if there are multiple nodes in $S_{i-1}$).

For each round $i\in [r,r']$ in the interval,
let $X_i$ be the random indicator variable that evaluates to $1$ if one of the following
conditions holds: (1) $|S|=n$; or (2) $|S|<n$ and $P_{u,t}(r,r')$ grows during the iteration corresponding
to round $i$.
Let $Y=\sum_{i\in [r,r']} X_i$.
Clearly, $Y$ is an upper bound on the size of $P_{u,t}(r,r')$ returned by the path construction procedure.
This upper bound is potentially loose, in that $Y$ could be much larger than the actual
path length; e.g., if we get to a state where $|S| = n$ early in the interval.
But it is sufficient for our purposes.

As we argued that for each $i$, $\Pr(X_i = 1) \geq 1/k$,
we can apply Lemma~\ref{lem:stochasticdominance} for these $z$ random variables, $p=1/k$, and $\varepsilon = 1/2$,
to derive:

\[\Pr\left[Y \leq (1/2)\frac{z}{k}\right] \leq \exp\left( - \frac{z}{8k} \right).  \]

Given our assumption that $z \geq 8xk\ln{n}$, this error bound is upper bounded by $\exp(-x\ln{n}) = n^{-x}$, as needed.
\end{proof}

%% file: lSmoothed.tex
\subsection{Random Broadcast with $\ell$-Smoothing}
\label{sec:Lsmoothing}
In this section, we look at the case where there are $\ell$ smoothed edges added per round.
Even when $\ell=1$, we get significant improvements, despite the fact that one edge only enables at most one additional
token be transferred per round; thus any non-trivial advantage conveyed
by smoothing does not come from directly increasing the capacity of
the underlying network. Our goal is to prove the following:

\begin{theorem}
  Fix any dynamic network ${\mathcal G}$ of size $n$. Fix any rumor set
  size $k \leq n$.  With high probability, random broadcast solves
  $k$-message broadcast in $O\left( \frac{kn^{2/3}\log^{1/3} n}{\ell^{1/3}}\right)$ rounds in ${\mathcal G}$
  with $\ell$-smoothing.
\label{thm:Lsmoothing}
\end{theorem}

Our proof proceeds in phases.  Note that the algorithm is the same
throughout, but our analysis focuses on different quality guarantees
in each phase. For the lemmas that follow, assume we
have fixed a dynamic graph ${\mathcal G}$, size $n$, and rumor set $T$ of
size $k$, as specified in the theorem.  

\subsubsection{Phase \#1: Spread}
The initial distribution of tokens is arbitrary, and tokens may be
located at very few nodes initially. The goal of this first phase is
to argue that after enough time tokens are spread out sufficiently across
the network, specifically spreading to a $\delta$ fraction of
nodes. The parameter $\delta$ is a function of $n$ and $k$ that we
shall set later to balance the length of all phases. 
We show here that $\Theta(k \delta n)$ rounds
suffice to spread all tokens to $\delta n$ nodes.  Notice that if $\delta = 1$, this
follows from Theorem \ref{thm:basic}.

\begin{lemma} \label{lem:spread}
  For any constant $x\geq 1$ and fraction $\delta$ with
  $(1/n)\ln{n} \leq \delta \leq 1$, there exists a constant
  $c_1 \geq 1$ such that with probability at least $1-n^{-x}$: for all
  $t\in T$, $n_t(R) \geq \delta n$, where $R = c_1 k \delta n$ is the
  duration of the phase.
\label{lem:1smoothing:step1}
\end{lemma}
\begin{proof}
  Fix a given token $t\in T$ and round $r \leq R$.  Let $X_r$ be the
  indicator random variable that evaluates to $1$ under two
  conditions: (1) $n_t(r) \geq \delta n$ (recall that $n_t(r)$ is the number of nodes that know token $t$ at the beginning of round $r$); or (2) a node
  learns $t$ for the first time during round $r$.  If the first
  condition does not hold then not all nodes know $t$ at round $r$, and hence there is at least one edge connecting
  a node $u$ that knows $t$ to a node $v$ that does not because the
  network is assumed to be connected. By the definition of random
  broadcast, $u$ selects $t$ to broadcast with probability
  $1/|T_u(r)| \geq 1/k$.  It follows that regardless of the execution
  through the first $r-1$ rounds: $\Pr(X_r = 1) \geq 1/k$.

  Let $Y = \sum_{r=1}^{R} X_r$.  We can apply our stochastic
  dominance result (Lemma~\ref{lem:stochasticdominance}) to $p=1/k$,
  $j = R$, and $\varepsilon = 1/2$ to derive the following:

\begin{eqnarray*}
 \Pr[Y \leq (1/2)\cdot(1/k) \cdot R] & \leq & 
\exp\left(-(1/8)\cdot \frac{R}{k}\right) \\
&=&\exp\left(-(1/8)\cdot c_1\delta n\right) \\
&\leq&
\exp\left(-(1/8)\cdot c_1 \ln(n) \right)\\
& = & n^{-c_1/8}.
\end{eqnarray*}

It is not hard to see that as long as $c_1 \geq 2$, then $Y \geq (1/2)\cdot(1/k) \cdot R$ implies that $n_t(R) \geq \delta n$.  This is because otherwise, it must be the case that every $X_r$ which is equal to $1$ is because a new node learned token $t$ at round $r$ (not because $n_t(r) \geq \delta n$).  But this implies that $n_t(R) \geq Y \geq (c_1 / 2) \delta n \geq \delta n$.

So if we set $c_1 = 8(x+1)$, we get that
the probability that $n_t(R) < \delta n$ is at most $n^{-(x+1)}$.  A union bound over all
$k\leq n$ tokens provides that $n_t(R) \geq \delta n$ for \emph{every} $t \in T$ with probability at least $1-n^{-x}$, proving the lemma.
\end{proof}

\subsubsection{Phase \#2: Seed}
Now that we have spread each token to $\Omega(\delta n)$ nodes, we next
rely on the edges added by smoothing to help sparsely seed these
tokens to new random nodes in the network.  In the third and final
phase we will show that this seeding is likely to have foiled the
adversary's attempts to keep certain nodes isolated from certain
tokens, and will have instead planted seeds sufficiently close on a
temporal path to arrive their destinations.  

In more detail, our goal is to show that, for some parameter $\gamma$, any 
sufficiently large set $S$ of size at least $\ln{n} / \gamma$ will have at 
least one node in the set receive a given token with high probability during the 
seed phase.  This phase will last for $\Theta((\gamma/\delta)kn)$
rounds following the conclusion of the spread phase.  

For example, consider $\delta = 1/k$ and $\gamma = 1/k^2$.  Then the length of both
the spread phase and the seed phase are $\Theta(n)$ rounds.  During
the seed phase with these parameters, each node has
probability of at least $1/k^2$ of receiving a specific token under
consideration.  Thus any set of size $k^2 \ln{n}$ will receive the 
token with high probability.  Note that these are not the best choices of $\delta$ and $\gamma$.  (With these choices, the total number of rounds including the sink phase is $O(n+k^3\log n)$ by Lemma~\ref{lem:Lsmoothcombine}.)

Our analysis of the seed phase focuses almost entirely on the 
smoothed edges added to the network topology graph in each round.  
\begin{lemma}
  Consider any constant $x\geq 1$ and positive fractions $\delta$ and $\gamma$.
  Fix any token $t \in T$, node set $S \subseteq V$ with
  $|S| \geq (1/\gamma) \ln{n}$, and round $r_0 \geq 1$ such that
  $n_t(r_0) \geq \delta n$.  With probability at least $1-n^{-x}$:
  $|S \cap \{ u \mid t\in T_u(r_0+2xR)\}| > 0$, where
  $R = (\gamma/\delta)kn/\ell > k\ln{n}$ and $2xR$ is the length of the phase.
\label{lem:seed}
\end{lemma}
\begin{proof}
  By assumption there are at least $\delta n$ nodes that know $t$ by the
  start of this phase (round $r_0$).  If any node in $S$ already knows
  $t$ at the beginning of this phase, then we are already done.
  So moving forward, assume no node in $S$ knows $t$ at the start of this
  phase.  Consider the set of potential edges that are useful, denoted
  $E_{\mathit{useful}}$, that would connect a node from the initial
  set that knows $t$ to a node in $S$.  It follows
  $|E_{\mathit{useful}}| \geq (\delta n) |S| \geq (\delta
  n)(\ln{n}/\gamma) = kn^2\ln(n)/(R\ell)$.  We now calculate, for a
  given round of phase $2$, the probability that a smoothed edge
  selected is from $E_{\mathit{useful}}$. To do so we leverage the
  specific definition of smoothing established in Section~\ref{sec:prelim} that treats this as a purely additive process: select a random edge
  from all possible edges; if it is not already present in the graph,
  add it to the graph; otherwise leave the graph the same.  Therefore,
  this probability is:

\[ \frac{|E_{\mathit{useful}}|}{\binom{n}{2}} > \frac{|E_{\mathit{useful}}|}{n^2} \geq (k\ln{n})/(R\ell).\]

As there are $\ell$ smoothed edges in a round, the probability that none of them are useful is:

\[(1 - k\ln{n}/(R\ell))^{\ell} \leq e^{-k\ln{n}/R} \leq 1 - k\ln{n}/(2R).\]

We call a round in phase $2$ {\em good} if an edge from
$E_{\mathit{useful}}$ is selected {\em and} the endpoint that knows $t$
selects $t$ to broadcast.  Because this latter selection event happens
with probability at least $1/k$, we can lower bound the probability of
a round being good as $p_{\mathit{good}} \geq \ln{n}/(2R)$.  If a round
is not good we call it {\em bad.}  Assume phase $2$ runs for $2xR$
rounds for constant $x>0$.  
Then the probability that every round is bad is bounded by:

\[ (1-p_{good})^{2xR} \leq \left(1-\frac{\ln{n}}{2R}\right)^{2xR} < \exp(-x \ln{n}) = n^{-x}, \]
as required by the lemma statement.
\end{proof}

\subsubsection{Phase \#3: Sink}
In the second phase, we leveraged smoothed edges to sparsely seed each
token throughout the network in a manner that is independent of the
adversary's construction of the dynamic graph.  In this final phase,
we deploy our predecessor path constructions to show it is likely for
each token $t$ and destination node $u$, that $t$ arrived at an
appropriate location in both time and topology to subsequently make
its way to $u$ like a flow heading toward a sink (hence the phase
name).  In particular, we simply need the final phase to be long
enough to achieve a predecessor path of size $\Omega(\ln{n}/\gamma)$.  Putting this piece together with the previous phases allows us to prove the following lemma, which almost immediately implies Theorem~\ref{thm:Lsmoothing}.

\begin{lemma}
  Fix any dynamic network $\mathcal G$ of size $n$.  Let $\delta$ and
  $\gamma$ be any values satisfying $(1/n)\ln{n} \leq \delta \leq 1$
  and $0 < \gamma \leq 1$.  Then with high probability, random
  broadcast completes $k$ message broadcast with $\ell$-smoothing in
  $O(k\delta n + (\gamma/\delta)kn/\ell + k\ln{n}/\gamma)$ rounds.
  \label{lem:Lsmoothcombine}
\end{lemma}
\begin{proof}
  Our analysis below makes use of a constant $x\geq 1$ that we will
  fix later.  Lemma~\ref{lem:1smoothing:step1} tells us that there
  exists a constant $c_1 \geq 1$, such that with probability least
  $1-n^{-x}$, for each token $t\in T$, at least $\delta n$ nodes know
  $t$ by round $c_1 k \delta n$, for some integer $c_1 >0$.  Let us
  call this the {\em spread condition.}  

  Before we apply the seed phase to each token, we need to first
  identify the sets we are attempting to seed. To do so, we look ahead
  to the sink phase.  Let $r_S$ indicate the first round of the sink
  phase, which we will calculate later based on the duration of the
  other two phases.  We run this phase for
  $z = \lceil 8xk\ln n / \gamma \rceil$ rounds.  That is, it runs between
  rounds $r = r_S$ and $r' = r_S + z$.  Theorem~\ref{thm:pathlength} tells us
  that with probability at least $1-n^{-x}$: for every node $u\in V$
  and token $t\in T$, the resulting predecessor path $P_{u,t}(r,r')$ is of size greater than $\frac{z}{2k}> \ln{n}/\gamma$.

  This allows to apply our seed phase (Lemma~\ref{lem:seed}) analysis
  to each such predecessor path.  This tells us that so long as the
  spread condition holds, for each such $S = P_{u,t}(r,r')$, the
  probability that some node in $S$ receives $t$ during the seed phase
  is at least $1-n^{-x}$.  The duration of the seed phase is
  $2x(\gamma/\delta)k n / \ell$ rounds.

  Finally, we note that by the definition of a predecessor path, if a node in $P_{u,t}(r,r')$
  receives a token $t$ by round $r_S$, then $u$ receives $t$ by round at most $r_S+z$.   We are left
  to pull together the pieces.  To do so, we note that success in
  solving $k$-message broadcast by the end of the sink phase requires
  the following events to occur:

\begin{enumerate}
    \item The spread condition holds. Call this $E_{spread}$
    \item For every $u\in V$ and $t\in T$, the predecessor
    path $P_{u,t}(r,r')$ is sufficient long. Call this: $E_{pred}(u,t)$.
    \item For each destination node $u\in V$ and token $t\in T$,
     at least one node in $P_{u,t}(r,r')$ receives $t$ during the seed phase.
     Call this $E_{seed}(u,t)$.
\end{enumerate}

By our above analysis the probability $E_{spread}$ fails is at most
$n^{-x}$, and the probability $E_{pred}(u,t)$ fails for {\em any} $u$
and $t$, as also upper bounded by $n^{-x}$.  For $E_{seed}(u,t)$, the
probability of failure for each specific pair $u$ and $t$, conditioned
on $E_{spread}$ and $E_{pred}(u,t)$, is upper bounded by $n^{-x}$.
Therefore, by a union bound, the probability that $E_{seed}$ fails for
any such pair is less than $n^{-(x+2)}$.  A final union bound then
provides that probability any of these events fail is less than
$n^{-x} + n^{-x} + n^{-(x+2)} < 1/n$ for a sufficiently large constant
$x$.

Plugging this constant value of $x$ into our above
round complexity bounds, and we get that random broadcast succeeds
with high probability in $c_1\delta kn + 2x(\gamma/\delta)kn/\ell + \lceil 8xk\ln n/\gamma\rceil
= O(\delta k n + (\gamma/\delta)kn/\ell + k\ln{n}/\gamma)$ total rounds,
as claimed in the theorem.
\end{proof}  

\begin{proof}[Proof (of Theorem~\ref{thm:Lsmoothing}).]
  Choose $\delta = (\log n / (n\ell))^{1/3}$ and $\gamma = (\ell^{1/3})(\log n / n)^{2/3}$.
  It follows $\gamma/\delta = (\ell^{2/3})(\log n / n)^{1/3}$.  Then by
  Lemma~\ref{lem:Lsmoothcombine} we get that broadcast completes in:
  
  \[O((\log n/(n\ell))^{1/3}kn + (\log n / n)^{1/3} kn/\ell^{1/3} + k\log{n} (n/\log n)^{2/3}/\ell^{1/3}) =
  O((kn^{2/3}\log^{1/3} n)/\ell^{1/3})\] 
  
  \noindent rounds, as claimed.
\end{proof}

\subsection{1-Smoothing}
The following result for $1$-smoothing is a simple corollary, which allows us to more easily compare to existing results.
\begin{corollary} \label{cor:1smoothing}
  Fix any dynamic network ${\mathcal G}$ of size $n$. Fix any rumor set
  size $k \leq n$.  With high probability, random broadcast solves
  $k$-message broadcast in $O\left(n + k^3 \log n\right)$ rounds in ${\mathcal G}$ with $1$-smoothing.
\end{corollary}
\begin{proof}
  Theorem~\ref{thm:Lsmoothing} implies that random broadcast takes at most $O(kn^{2/3} \log^{1/3} n)$ rounds.  It is easy to see that $kn^{2/3} \log^{1/3} n$ is at most $n+k^3 \log n$ for all values of $k$.  Alternatively, we can set $\delta = 1/k$ and $\gamma = 1/k^2$, and apply Lemma~\ref{lem:Lsmoothcombine} with $\ell = 1$.  
\end{proof}

%% file: lower.tex
\subsection{Lower Bound for Random Broadcast in Smoothed Networks}
\label{sec:lower}

In this section we prove the following theorem.

\begin{theorem} \label{thm:lower}
For all $n,k,\ell \geq 1$, there are dynamic networks on $n$ nodes and a starting token distribution of $k$ tokens such that random broadcast with $\ell$-smoothing has expected completion time of at least $\tilde \Omega\left( \min\left( \frac{kn^{2/3}}{(\ell(k+\ell))^{1/3}}, \frac{kn}{k+\ell}\right)\right)$.
\end{theorem}

In most ``reasonable'' regimes ($k$ and $\ell$ not overwhelmingly large) the minimum in the above lower bound will be achieved by $\frac{kn^{2/3}}{(\ell(k+\ell))^{1/3}}$.  Note that this almost matches the upper bound of Theorem~\ref{thm:Lsmoothing}: it is off by just a $(k + \ell)^{1/3}$ factor.  In addition, we also note that when $\ell$ is large the upper bound \emph{cannot} be tight, while our lower bound can be.  To see this, consider the case where $\ell = n$ and $k = o(n)$.  For these parameters, essentially every node is the endpoint of an edge added by smoothing in every round, and for each token the probability of broadcasting it is at least $1/k$, and hence for every token the number of nodes who know it will double in at most $O(k)$ rounds.  Thus random broadcast  will complete after only $O(k \log n)$ rounds.  For this case, where $\ell = n$ and $k = o(n)$, our lower bound from Theorem~\ref{thm:lower} correctly reduces to $\tilde \Omega(k)$, while the upper bound from Theorem~\ref{thm:Lsmoothing} remains at $\tilde O(k n^{1/3})$.
This hints that the modest gap between our upper and lower bounds might ultimately be resolved to be closer to the latter result.


\subsubsection{Proof of Theorem~\ref{thm:lower}}

Our lower bound instance will be the \emph{dynamic star}.  The vertices are $v_0, \dots, v_{n-1}$, and at time $i \in \{1,2,\dots, n\}$ the graph will be a star with $v_i$ at the center.  Initially node $v_0$ knows all of the tokens, while nodes $v_1, \dots, v_{n-1}$ knows all of the tokens \emph{except} for token $1$.  So random broadcast is complete once all nodes know token $1$.  Note that this graph is not defined for more than $n$ rounds, but the lower bound that we are trying to prove is at most $n$ due to the second term in the min, so we will not need to consider rounds past $n$.

Let $t = \min\left( \frac{kn^{2/3}}{(\ell(k+\ell))^{1/3}}, \frac{kn}{k+\ell}\right) / (2000\log n)$ (we have not optimized the constant or log factors).  We will argue that with constant probability, random broadcast has not completed by time $t$.

Let $A = \{v_i : 1 \leq i \leq t\}$.  For $v_i \in A$, we say that $v_i$ is a \emph{good} node if, conditioned on $v_i$ knowing token $1$ in round $i$, $v_i$ will broadcast token $1$ in round $i$ (the round where $v_i$ is the center).  Let $B$ denote the set of good nodes.   Let $C_i$ denote the set of nodes who know token $1$ at the beginning of round $i$.  Without smoothing we would have that $C_i \subseteq \{v_0, v_1, \dots, v_{i-1}\}$, but with smoothing this is not necessarily true.  We begin by analyzing $C_i$ under a condition on the good nodes.

\begin{lemma} \label{lem:Ci}
Suppose that for every $i \leq t$, either $v_i$ is not a good node or $v_i$ does not know token $1$ at the beginning of round $i$.  Then $|C_i| \leq 100 i \left(\frac{k+\ell}{k}\right) \log n$ for all $i \leq t$ with high probability.
\end{lemma}
\begin{proof}
By assumption, for all $i \leq t$ the center node of the star does not broadcast token $1$.  Hence in round $i$, the nodes who learn token $1$ consist of at most the center node and some other nodes who learn token $1$ via smoothed edges.  Even if \emph{all} of the $\ell$ smoothed edges had an endpoint in $C_i$, the expected number of them who transmit token $1$ is at most $\ell /k$.   Hence a Chernoff bound implies that $|C_i| \leq |C_{i-1}| + 100\log n \cdot \left(1 + \frac{\ell}{k}\right)$ with high probability.  This high probability allows us to do a union bound over the first $t$ rounds, implying that $|C_i| \leq 100 i \left(\frac{k+\ell}{k}\right) \log n$ with high probability. 
\end{proof}

We say that round $i$ has \emph{productive smoothing} if some node in $B$ learns token $1$ via an edge added by smoothing.  It is easy to see that if after $t$ rounds there has not been any round with productive smoothing, then the assumption in Lemma~\ref{lem:Ci} holds, and hence $|C_i| \leq 100 i \left(\frac{k+\ell}{k}\right) \log n$ for all $i \leq t$ with high probability.  Since $t <  \frac{nk}{100(k+\ell)\log n}$ this means that not all nodes know token $1$ at time $t$, and so random broadcast has not finished.  So we just need to argue that with constant probability, there have been no rounds with productive smoothing before round $t$.

To see this, first note that by the definition of random broadcast, each node in $A$ is good independently with probability $1/k$.  Hence the expected number of good nodes is $|A| / k = t/k$.  A standard Chernoff bound then implies that $|B| \leq (10 t / k) \log n$  with high probability, so from now on we will condition on this being true.

In order for round $i$ to have productive smoothing, at least one of the $\ell$ edges added by smoothing must have one endpoint in $C_i$, one endpoint in $B$, and the endpoint in $C_i$ must choose to broadcast token $1$.  The probability of this for a single random edge is $\frac{|C_i|}{n} \cdot \frac{|B|}{n} \cdot \frac{1}{k}$, and hence a union bound over all $\ell$ random edges added in rounded $i$ implies that the probability of productive smoothing in round $i$ is at most
\[
\frac{|C_i|}{n} \cdot \frac{|B|}{n} \cdot \frac{\ell}{k} \leq \frac{|C_t|}{n} \cdot \frac{10 t \log n}{kn} \cdot \frac{\ell}{k} = |C_t| \cdot \frac{10t\ell \log n}{(kn)^2}.
\]

If there has been no productive smoothing before round $i$, then Lemma~\ref{lem:Ci} implies that with high probability $|C_i| \leq 100 i \left(\frac{k+\ell}{k}\right) \log n$.  This is high enough probability for us to take a union bound and still have a high probability bound, so we will assume that \emph{if} there has been no productive smoothing before round $i$ then $|C_i| \leq 100 i \left(\frac{k+\ell}{k}\right) \log n$.  

Let $X_i$ be an indicator random variable for the event that round $i$ has productive smoothing.  Then
\begin{align*}
    \Pr[\exists i \in [t] : X_i = 1] &\leq \sum_{i=1}^t \Pr\left[X_i = 1 \mid \sum_{j=1}^{i-1} X_j = 0\right] 
    \leq \sum_{i=1}^t \left(100 t \left(\frac{k+\ell}{k}\right) \log n\right)\left(\frac{10t\ell \log n}{(kn)^2}\right) \\
    &=\frac{1000t^3 \ell \left(k+\ell\right)\log^2 n}{k^3 n^2}
\end{align*}
Since $t \leq \frac{k n^{2/3}}{(2000 \ell (k+\ell) \log^2 n)^{1/3}}$ then this probability is at most $1/2$, which (as discussed) implies Theorem~\ref{thm:lower}.

%% file: static.tex
\section{Random Broadcast in Smoothed Static Networks}
\label{sec:static}

As previously argued,
 a minimum amount of smoothing (i.e., $\ell=1$)
 improves the performance of random broadcast in dynamic networks from $O(kn)$ to $\tilde{O}(kn^{2/3})$.
 To better understand how smoothing supports information spreading,
 a natural follow-up question is to investigate its impact on random broadcast in {\em static} networks.
 If smoothed analysis provides the same bounds for both the dynamic and static settings,
 this would imply that smoothing essentially bypasses the difficulties induced by changing graph edges.
 Here we show this is not the case. 
 In more detail,
we prove that in static networks, $1$-smoothing improves the complexity of random broadcast
down to $\tilde{\Theta}(k\sqrt{n})$ rounds, beating what we can guarantee in the dynamic setting.
This establishes a gap between static and dynamic networks with respect to random broadcast,
confirming the intuition that network dynamism introduces unique difficulties for information dissemination
that smoothing alone cannot fully overcome.

\subsection{Upper Bound}
We begin by upper bounding the performance of random broadcast when solving $k$-message broadcast
in a static graph of size $n$.
We prove that in this setting random broadcast  solves the problem in $O(k\sqrt{n}\log^2{n})$ rounds,
with high probability.
Critical to our analysis is the following graph decomposition result:

\begin{lemma}
Fix a connected static graph $G=(V,E)$ of size $n$.
There exists a partition of $V$ into components $C_1,C_2,\ldots,C_x$,
such that for each $i$, $1 \leq i \leq x$: (1) $|C_i| \geq \sqrt{n}$; (2) the subgraph of $G$ induced by $C_i$ is connected and has a diameter at most $6\sqrt{n}$.
\label{lem:decomp}
\end{lemma}
\begin{proof}
We prove the lemma constructively.
Given a static connected graph $G=(V,E)$ of size $n$, we describe and analyze a two-stage iterative procedure that
constructs components that satisfy the desired properties.

During the first stage, we partition $V$ into {\em preliminary} components, $\hat C_1, \hat C_2, \ldots$,
and provide each $\hat C_i$ a color $c_i$ from $\{$red, blue$\}$.
To do so, we begin by first labelling each node in $V$ as free.
We then proceed in construction phases, labelled $1,2,\ldots$, until no free nodes remain.
In more detail, for each phase $i$:

\begin{enumerate}
    \item Select an arbitrary free node $u$.
    \item In the subgraph of $G$
    induced by nodes that remain free, conduct a breadth-first search, starting from $u$,
    that terminates when either: (1) the search has encountered at least $\sqrt{n}$ free nodes;
    or (2) the search gets stuck with no further free nodes to explore.
    
    \item Define $\hat C_i$ to contain every node reached in the search.
    \item If the search from step 2 ended due to criteria (1), set $c_i =$ red; else
     if the search ended due to criteria (2), set $c_i = blue$.
    
\end{enumerate}

By definition, the resulting preliminary components partition the nodes in $V$.
We will use these preliminary components to construct the final components that satisfy the lemma statement.
To do so, we first note that each blue preliminary component must neighbor at least one red preliminary component.
To see why, assume for contradiction that some blue component $\hat C_i$ neighbors only other blue components.
Fix one such neighboring component $\hat C_j$.
Because they are neighboring, we can fix some $u$ and $v$ such that $u\in \hat C_i$, $v\in \hat C_j$, and
$\{u,v\}$ is in $G$.
Assume without loss of generality that $i<j$.
When the construction of $\hat C_i$ terminates, $v$ is still free, as it is included in the construction of a later
component.
By assumption, $c_i =$ blue, meaning that the construction of $\hat C_i$ terminated because it could find no further free nodes to explore.
We just established, however, that $v$ was still free and connected to $\hat C_i$ during the round when it terminates---a contradiction.

Given this observation, we can iterate through the blue preliminary components, 
combining the nodes in each with a neighboring red preliminary component. 
Let $C_1,C_2,\ldots,C_x$, be the components that result after we finish these merges.
Since each $C_i$ begins with a red preliminary component, we get that $|C_i| \geq \sqrt{n}$ as required.
We next turn our attention to the diameters of these components.
We first note every preliminary component has a diameter bounded by $2\sqrt{n}$,
as in both cases the breadth-first search tree defining the component has a height bounded by $\sqrt{n}$.
Next, fix some component $C_i$, and two unique nodes $u,v\in C_i$.
If they are in the preliminary component, then they are within $2\sqrt{n}$ hops.
If they are are in two different preliminary components, then they are at most $6\sqrt{n}$ hops
from each other, as $2\sqrt{n}$ hops gets you to the preliminary red component at the core of $C_i$,
an additional $2\sqrt{n}$ hops gets you across the red core to any other neighboring preliminary component,
and a final $2\sqrt{n}$ hops gets you to any node in this new preliminary component.
\end{proof}

We are now ready to prove our upper bound.
The key intuition in the following argument is that we can we analyze the spread of a given
target token within the context of the components provided by Lemma~\ref{lem:decomp}.
We will show, roughly speaking, that when the target token is first spreading,
if ${\cal A}$ is the set of components that known the target, 
it is likely that within $\tilde{O}(k\sqrt{n})$ rounds,  each component
in ${\cal A}$ will succeed in seeding the target over a smoothed edge to a unique
component not in ${\cal A}$, effectively {\em doubling} the number of components
that have learned the token. A logarithmic number of such doublings are sufficient
to spread the target to at least half the network, 
at which point the analysis shifts to the perspective of the remaining components,
and argues that within an additional $\tilde{O}(k\sqrt{n})$ rounds, 
each is likely to connect to an already informed component by a smoothed edge
and receive the token.
Care is needed in the formal argument to deal with both uneven-sized components
and dependencies between the smoothed edge behavior in different components during
the same spreading interval.

\begin{theorem}
Fix some connected static graph $G$ of size $n$.
Fix any rumor set size $k \geq 1$.
With high probability, random broadcast solves $k$-message broadcast in $O(k\sqrt{n}\log^2{n})$ rounds in $G$ with $1$-smoothing.
\label{thm:static}
\end{theorem}
\begin{proof}
Fix a graph $G= (V,E)$ of size $n$, and a rumor set of size $k$, as specified by the theorem statement.
Fix an arbitrary {\em target token} from the rumor set.  
We argue that this target token will spread to full network with the stated complexity.

To do so, we first apply Lemma~\ref{lem:decomp} to partition $G$ into components $C_1,C_2,\ldots,C_x$ with
the specified properties.
Moving forward, in a given round,
we call a component {\em seeded} if at least one node in the component
knows the target token, and call it {\em completed} if every node
knows the target.
It is straightforward to establish that once a component is seeded,
it will be completed, with high probability, in an additional $O(k\sqrt{n}\log{n})$
rounds. One approach to making this argument is to fix a breadth-first search tree in the component rooted at a node that knows the target token.
This root will broadcast the target in each round with probability at least $1/k$. Because each broadcast decision is independent,
we can apply a Chernoff bound to establish that in $O(k\log{n})$ rounds,  the root will broadcast the target at least once, with high probability.

We can then repeat this analysis to show that within an additional
$O(k\log{n})$ rounds, every node at level $1$ in the tree will have sent the target token, informing every node at level $2$, applying union bounds to ensure that every node at level $1$ succeeds with sufficient probability. An additional $O(k\log{n})$ rounds moves the token to level $3$, and so on. 
By the guarantees of Lemma~\ref{lem:decomp}, the tree has $O(\sqrt{n})$ levels, meaning that $O(k\sqrt{n}\log{n})$ rounds is sufficient to spread a target token through the whole component
with high probability.

Returning to our main argument, fix a round $r$.
Let ${\cal A}_r$ be the subset of components 
that are seeded at the beginning of
round $r$.
Let $n_r$ be the number of nodes in components in ${\cal A}_r$.
We first consider the case where $n_r < n/2$.
To do so, let ${\cal B}_r$ be all the components not in ${\cal A}_r$.
Our goal is to show that in an additional $O(k\sqrt{n}\log{n})$ rounds,
the target token will be delivered over smoothed edges
to a collection of components in ${\cal B}_r$,
where they will then spread to a total number of new nodes that is
at least a constant fraction of $n_r$---increasing the number
of nodes that know the target token by a constant factor.

To make this argument, 
we divide this period of $O(k\sqrt{n}\log{n})$ rounds starting at round $r$
into three parts.
During the first part, 
we wait for the target token to spread sufficiently
to complete every component in ${\cal A}_r$.
As argued above, with high probability, this takes $O(k\sqrt{n}\log{n})$ rounds.

Next we consider a stretch of an additional
$O(k\sqrt{n}\log{n})$  rounds that we call the {\em seeding interval}.
We will show that during this interval, smoothed edges
between components in ${\cal A}_r$ and ${\cal B}_r$,
will seed the target token into a collection of ${\cal B}_r$ components
whose collective size is a constant fraction of $n_r$.
The final $O(k\sqrt{n}\log{n})$ rounds will be dedicated
to allowing these seeded ${\cal B}_r$ components to complete.
(Of course, it is possible that in the time we spent completing components
in ${\cal A}_r$, the target token might have already made its way to components in ${\cal B}_r$ and started spreading, but this only speeds up our efforts.)

We are left then to study closely the behavior of the smoothed edges during
the smoothing interval of this period.
Our attempts to analyze smoothed edges from ${\cal A}_r$ to ${\cal B}_r$
during this  interval
will be complicated by the fact that the components in ${\cal A}_r$ can be of variable size. Though
Lemma~\ref{lem:decomp} guarantees that each contains at least $\sqrt{n}$ nodes, it is possible that some might be much larger than this lower limit.
It is useful for our purposes to temporarily reorganize these already informed nodes into more uniform-sized groups.
With this in mind, let ${\cal A'}_r$ be a partition of the nodes in components in ${\cal A}_r$ into {\em groups} that are all of size $\Theta(\sqrt{n})$.
For our purposes, it does not matter how this partition is defined. The nodes in each $S\in {\cal A'}_r$, for example, do not need to be connected. In the next step of our analysis, we will only concern ourselves with the probability that a node in a given group is selected as an endpoint of a smoothed edge.

Next, fix some arbitrary group $S_1\in {\cal A'}_r$ to consider first.
Let $\hat n_r$ be the number of nodes in ${\cal B}_r$.
Because we are still considering the high-level case where  $n_r < n/2$, we know $\hat n_r \geq n/2$.
We now analyze the rounds of the seeding interval in order,
stopping at the first round in which: (1) a smoothed edge connects
a node $u\in S_1$ to a node in a component in ${\cal B}_r$; and
(2) $u$ broadcasts the target token.
It is straightforward to see that in any given round,
a {\em productive} connection of this type occurs
with probability at least:

\[ \frac{|S_1|}{n}\cdot \frac{\hat n_r}{n} \cdot \frac{1}{k} \geq 
     \frac{1}{2k\sqrt{n}}.\]
     
With high probability, therefore,
we will find such a productive connection in a seeding
interval of length $O(k\sqrt{n}\log{n})$.
Where this argument gets more subtle is that we now 
want to consider the other groups in ${\cal A'}_r$,
and argue that they too will succeed in forming
a productive connection during this {\em same} seeding
interval.
This will require care to deal properly with dependencies.

The first thing we do is take the component in ${\cal B}_r$
at the other end of the productive connection from $S_1$,
and add it to a set ${\cal C}$  of successfully seeded
components.
Before proceeding in our analysis of this seeding interval, 
we make two checks to see if we are already done.
Let $n'$ be the number of nodes in components in ${\cal C}$ at this point.
If $n_r + n' \geq n/2$, then we can simply
wait an additional $O(k\sqrt{n}\log{n})$ rounds to complete
the component in ${\cal C}$, and be done with the high-level case 
we are considering in which less than half of the nodes know the target token.
Similarly, if $n' \geq n_r/2$,
then we can finish our analysis of this particular seeding interval as we
have accomplished our proximate goal of increasing the number
of nodes that know the target token by a constant factor.

If we fail both checks, we must then continue with analyzing our
same seeding interval in the hopes of seeding more tokens into ${\cal B}_r \setminus {\cal C}$.
Fix a new group $S_2\in {\cal A'}_r \setminus S_1$.
As before, consider rounds in the seeding interval one by one,
starting from the first round of the interval,
until
we arrive at a with a productive connection
from $S_2$
to a not yet seeded component in ${\cal B}_r \setminus {\cal C}$.
In each such round, we argue that the probability of a productive
connection is still in $\Omega(\frac{1}{k\sqrt{n}})$.
The main difference as compared to prior groups considered
is that the number of nodes that can receive a productive
smooth edge decreases as we add more components to ${\cal C}$.
By our above check, however, we know that the number
of nodes in ${\cal C}$ is less than $n_r/2 < n/4$.
Because we similarly assume that ${\cal B}_r$ has at least
$n/2$ nodes, then ${\cal B}_r \setminus {\cal C}$ must still
have at least $n/4$ total nodes.
The probability of selecting
an endpoint in ${\cal B}_r \setminus {\cal C}$ will therefore
always be at least $1/4$, reducing the original productive connection
probability calculated for $S_1$ by only a constant factor for later groups.

As previewed, however, we must also consider dependencies.
When considering a given round $r'$ in the seeding interval
when considering group $S_2$, there are two relevant possibilities:
(1) $r'$ was the productive connection from our analysis of $S_1$;
(2) $r'$ was not the productive connection for $S_1$.
The first case introduces a problematic dependency, as
being productive for one group prevents you from being productive
for another.
To deal with this dependency we simply {\em ignore} in our analysis
any round that was part of a productive connection for a previously
studied group.
The second case, by contrast, introduces a useful dependency:
knowing that $r'$ was not productive for $S_1$ only {\em increases}
the probability that it is productive for $S_2$.
A standard negative correlation argument tells us that when lower bounding
the probability of a productive connection with respect to $S_2$,
it is fine to treat each round in the second case as succeeding with an independent
probability in $\Omega(\frac{1}{k\sqrt{n}})$.

It follows that with high probability, $S_2$ will also have a productive connection in our seeding interval. At this point, we can repeat the above analysis.
Add the newly seeded component to ${\cal C}$. 
Check if we are done.
If not, select a new set $S_3\in {\cal A'}_r\setminus \{S_1, S_2\}$ 
and study the same interval again, round by round, 
ignoring now only the two rounds in which $S_1$ and $S_2$ succeeded in forming productive connections,
and so on, until we finally match the criteria for stopping our analysis.
(Notice that removing productive rounds from consideration
in the seeding interval is not a problem as there are at most $O(\sqrt{n})$ such productive rounds possible given that $|{\cal A'}_r| = O(\sqrt{n})$,
and our interval length can be made sufficiently long to still provide
us the needed high probability of success even with up to $O(\sqrt{n})$ rounds
omitted from consideration.)

Moving on with our argument,
recall that for our fixed round $r$, 
there are two different criteria that might terminate the above seeding analysis.
The first is that at least half the nodes in the network learn the target token, at which point we are ready to move on to the second half of our overall analysis, which we will discuss shortly.
The second termination criteria is that the number of nodes in components
in ${\cal B}_r$
that learn the target token is at least a constant factor of $n_r$.
In this case, we fix a new round $r'$ after the seeding interval in question is done, and after all components in ${\cal C}$ complete.
We then reapply our above analysis starting from $r'$, increasing the number
of nodes that know the target node by another constant factor.
We can repeat this at most $O(\log{n})$ times before at least half the nodes know the target token.

We now consider what happens
once we succeed in spreading the target token to at least half the nodes in the network. 
We can now redeploy pieces of our above spreading argument to show that target token will make it to all remaining nodes
in just one more additional spreading interval of length $O(k\sqrt{n}\log{n})$ rounds.

We first dispense with the sub-case in which less than $\sqrt{n}$ nodes do not the token; i.e., we are almost done. If this is true, each uninformed node $u$ is within $\sqrt{n}$ of at least one informed node, meaning a straightforward spreading analysis will deliver the token to each such $u$ with high probability
within $O(k\sqrt{n}\log{n})$ rounds---completing the rumor spreading.

We are left with the sub-case in which somewhere between $\sqrt{n}$ and $n/2$ nodes remain
that do not have the target token.
Let ${\cal A}$ be the set of components that contain at least one node that from this set
of nodes that do not know the target.
Some of these components are seeded (i.e., at least one node in them knows the target) and some are not (i.e., no node knows the token).

Let ${\cal A'}$ be the subset of ${\cal A}$ that contains only the non-seeded components.
Spend $O(k\sqrt{n}\log{n})$ to complete the seeded components from ${\cal A}$.
We now turn our attention exclusively to the components from ${\cal A'}$,
as these are the only components at this point which can possibly contain {\em any}
nodes that do not know the target token.
We know each such component is of size at least $\sqrt{n}$,
and that at least half the total nodes in the network are not in
these components.
We can therefore apply our above spreading analysis to the components in ${\cal A'}$,
which establishes that in a single additional spreading interval,
every component in ${\cal A'}$ will have a productive connection with a completed component, thus seeding it with the target token. We can then complete these components,
and therefore complete $k$-message broadcast,
in an additional $O(k\sqrt{n}\log{n})$ rounds.

To conclude the proof,
we note that all relevant growth and spreading arguments hold
with high probability, and that there are at most $\text{poly}(n)$ such
events that must succeed. This allows us to deploy union bounds
to prove that the entire problem is successful, with high probability,
after $O(\log{n})$ intervals of length $O(k\sqrt{n}\log{n})$,
yielding the claimed overall time complexity of $O(k\sqrt{n}\log^2{n})$ rounds.
\end{proof}

\subsection{Lower Bound}

We now prove that our analysis of random broadcast in static networks is tight (within logarithmic factors).
The argument formalized below is easy to summarize. Consider spreading a target token through a line in which every node
already knows $\Theta(k)$ other tokens.
It will require $\Theta(k\sqrt{n})$ rounds for the target token to directly spread through the first $\sqrt{n}$ nodes in the line
with reasonable probability. At the same time, this interval is not sufficiently long for smoothed edges to have a reasonable probability of helping to speed up this initial spread.
Formally:

\begin{theorem}
Fix a network size $n > 1$ and rumor set size $k>1$.
The expected time for random broadcast to solve $k$-message broadcast in networks of size $n$ with $1$-smoothing is in $\Omega(k\sqrt{n})$.
\label{thm:static:lower}
\end{theorem}
\begin{proof}
Fix any $n$ and $k$ as specified.
Assume that the tokens are labelled from $1$ to $k$, 
and that random broadcast makes its token selection
at each node by choosing a random value from $1$ to $k$,
and then broadcasting the token with the label closest to random value.
It is straightforward to show by a simulation relation
argument that a lower bound that holds for this stronger
algorithm still holds for the standard version in which the algorithm selects a token uniformly from those known by the node.
We consider this stronger version here as it simplifies the presentation of our lower bound argument.

Consider a graph $G$ in which $n$ nodes are connected in a line, arranged $u_1,u_2,\ldots,u_n$.
Fix a rumor set of $k$ tokens. Assign every node tokens $2$ through $k$.
Start token $1$ only at node $u_1$.
Let $Y$ be the random variable that describes the number
of number of rounds required for token $1$ to make it to all
nodes in the set $\{u_1,u_2,\ldots,u_{\sqrt{n}}\}$ using
only the random broadcast mechanism (i.e., ignoring smoothed edges).
We can define $Y=X_1 + X_2 + \ldots + X_{\sqrt{n}-1}$, where $X_i$ is
the number of rounds until $u_i$ first selects value $1$ after
token $1$ first arrives at $u_i$ via the random broadcast mechanism.
Formally, let $X_1$ be the first round in which $u_1$ selects value
$1$; and for $i>1$, let $X_i$ be the number of rounds until $u_i$ first selects
$i$ after round $X_{i-1}$.
By linearity of expectation, $E(Y) = \Theta(k\sqrt{n})$.

Next we consider smoothed edges.
We call a round $r$ {\em useful} if: (1) the smoothed
edge in $r$ includes at least one endpoint among the first $\sqrt{n}$
nodes; and (2) the endpoint(s) from among the first
$\sqrt{n}$ nodes randomly select value $1$.
The probability that a given round is useful is in $O(1/(k\sqrt{n}))$.

To obtain our bound,
let $t=\alpha k \sqrt{n}$,
where $\alpha > 0$ is a fraction that is sufficiently small to ensure:
(1) the expected number of useful rounds in the first $t$ rounds
is less than $1$; and (2) $t < E(Y)$.
It follows that with constant probability, during the first $t$ rounds:
(1) there are no useful smoothed edges (which could potentially help token $1$ spread faster);
and (2) token $1$ does not make it to node $u_{\sqrt{n}}$ via the 
random broadcast mechanism.
If both of these events are true, then 
 $k$-message broadcast is not solved in $t$ rounds.
It follows that when calculating the expected complexity
of $k$-message broadcast in this network for this initial
token assignment, there is constant probability mass
dedicated to complexities of magnitude at least $t$,
meaning that the expected overall complexity must be in $\Omega(t) = \Omega(k\sqrt{n})$, as claimed.
\end{proof}

%% file: wellmixed.tex
\section{Random Broadcast in Well-Mixed Networks}
\label{sec:wellmixed}

Dutta et al.~\cite{dutta2013complexity} introduced the notion a {\em well-mixed} network in the context
of the $k$-message broadcast problem. 
A network satisfies this property if for each token $t$ and node $u$,
with some independent constant probability: $u$ starts with $t$.
The main $\tilde{\Omega}(nk)$ lower bound from this paper also holds for well-mixed networks.
To circumvent this bound,
the authors assume a stronger communication model that allows independent interactive
communication on each  edge, and provide a randomized
algorithm in this setting that solves $k$-message broadcast in
$\tilde{O}(n + k)$ rounds.

Here we deploy our predecessor path constructions
from Section~\ref{sec:pred}, originally designed to
support our smoothed analysis, to now  explore an alternative method
to circumvent the $\tilde{\Omega}(nk)$ lower bound without smoothing: weakening the adversary.
The lower bound in~\cite{dutta2013complexity} assumes a {\em strongly
adaptive} adversary that knows all the nodes' random bits, 
allowing it to generate a network graph in each round based on
the specific tokens nodes will broadcast during that round.
Another natural option is the {\em oblivious} adversary assumed
in this paper, in which the adversary generates the graph 
without advance knowledge of the nodes' random bits.
Indeed, the question of whether an oblivious adversary enabled better
bounds was identified as important future work in~\cite{dutta2013complexity}.

Using predecessor paths, we prove that with an oblivious adversary,
our simple random broadcast strategy solves $k$-message broadcast
in well-mixed networks in $\tilde{O}(k)$ rounds.
This improves on the $\tilde{\Omega}(n+k)$ bound from~\cite{dutta2013complexity},
as it eliminates the $n$ factor.\footnote{It may seem at first odd that the bound does not include $n$, as even in a static
line it takes at least $n$ rounds for a single token to make it to all nodes. 
In the well-mixed model, however,
that token would be expected to show up every constant number of hops on average, and therefore never be too far from any particular destination. Our analysis uses predecessor paths to
generalize this observation to dynamic networks
and show that each token starts
on a node that is
not too far in space and time for each given destination.} 
It also matches the trivial $\Omega(k)$ lower bound that holds
for {\em all} $k$-message
broadcast algorithms in a well-mixed network.\footnote{Fix a static line.
Let $u$ be one of the endpoints. With constant probability $u$ is missing
$\Omega(k)$ tokens in its initial set. Because $u$ can receive at most
one token per round in a line topology, it requires at least
$\Omega(k)$ rounds
for it to learn these missing tokens one by one.}
Indeed, our result is actually more general, showing $\tilde{O}(k/p)$
rounds are needed, when $p$ is the token probability; i.e., the definition
of well-mixed in~\cite{dutta2013complexity} assumes $p=\Theta(1)$.

This result opens a clear separation between strongly adaptive and oblivious
adversaries in the context of well-mixed networks, and hints such a separation
might exist for arbitrary token distributions as well.
It also provide further evidence that the simple random broadcast strategy
is a highly effective
strategy for information dissemination in these settings.

Formally, we consider the following generalized version of the property concerning initial token distributions introduced in~\cite{dutta2013complexity}:  

\begin{definition}
Fix a probability $p>0$.
We say an initial token distribution is {\em $p$-mixed} if for each
node $u\in V$ and token $t\in T$: $u$'s initial token set includes $t$
with independent probability $p$.\footnote{A slight technicality of this definition is that as stated
it allows for certain tokens to not show up at all, making termination impossible. A simple
fix is to assume that all $k$ tokens are distributed arbitrarily to at least one node, then
the random process is deployed to further introduce more tokens into the system.}
\end{definition}

Given this definition, we prove our main upper bound result:

\begin{theorem}
Fix a probability $p>0$. 
Random broadcast solves $k$-message broadcast in $O((k/p)\log{n})$ rounds, with high probability,
when run in a network with a $p$-mixed token distribution.
\label{thm:pmixed}
\end{theorem}
\begin{proof}
Fix any $u$ and $t$.
Our first step is to choose a large enough $r$ such that $|P_{u,t}(1,r')| \geq \alpha(1/p)\ln{n}$,
for a constant $\alpha \geq 1$ we will fix later.
We want this to hold with probability at least $1-n^{-4}$.
To do so, we can apply Theorem~\ref{thm:pathlength} with $x=4$, $r=1$,
$r'=32k\alpha(1/p)\ln{n}$.
(Notice, for our parameters,
the corresponding $z=r'-r$ value is lower bounded by $8xk\ln{n} = 32k\ln{n}$, as required by the theorem statement.)
This tells us that with probability at least $1-n^{-4}$,
$|P_{u,t}(1,r')| > \frac{32k\alpha(1/p)\ln{n}}{2k} \geq \alpha(1/p)\ln{n}$, as needed.
A union bound over the $nk \leq n^2$ possible combinations of processes and tokens,
tells us that the probability that any predecessor path is less than this length is less than $n^{-2}$.

Fix some such path $P_{u,t}(1,r')$ that is sufficiently long;
i.e., $q = |P_{u,t}(1,r')| \geq \alpha(1/p)\ln{n}$.
We apply the definition of $p$-mixed to argue that it is likely that at least one node
in this path starts the execution with token $t$.
By definition, of a predecessor path, each node in $P_{u,t}(1,r')$ is unique.
By the definition of $p$-mixed, each $u_i$ starts the execution with token $t$
with independent probability $p$.
Given this independence, we can use the indicator random variable $X_i$ to describe
whether or not $u_i$ starts with token $t$,
and then apply our Chernoff Bound form from Theorem~\ref{thm:chernoff} to $X=\sum_{i=1}^{r'}X_i$,
to bound the probability that $X$ is far from its expectation $\mu=pq$.
Formally:

\[ \Pr[Y \leq \mu/2] \leq \exp(- \mu/8).\]

\noindent Given that $\mu = pq \geq \alpha\ln{n}$, then for $\alpha\geq 32$,
it follows:

\[ \Pr[Y = 0 ] < n^{-4}. \]

\noindent A union bound over the $nk<n^2$ node and token pairs tells us that if at all predecessor
paths are of length at least $q$, the the probability {\em any} path does not contain at least one
node starting with the path's target token is less than $n^{-2}$.
A union bound over the failure probabilities for these two events gives us the final result
that with high probability, for every $u\in V$ and $t\in T$,
at least one node in $P_{u,t}(1,r')$ starts with $t$.
As argued in our previous discussion of predecessor paths, if
any node on $P_{u,t}(1,r')$ starts with $t$ then $u$ receives $t$ by round $r'$.
Therefore, with high probability, random broadcast solves $k$-message broadcast
in a $p$-mixed network in $r' = O((k/p)\log(n))$ rounds, as claimed.
\end{proof}